\documentclass[a4paper, amsfonts, amssymb, amsmath, reprint, showkeys, nofootinbib, twoside,notitlepage,onecolumn]{revtex4-2}

\def\apj{{ApJ}}                 

\def\prd{{Phys.~Rev.~D}}        
\def\prl{{Phys.~Rev.~Lett.}}    
\usepackage{amsmath,amstext}
\usepackage[T1]{fontenc}
\usepackage{amssymb}
\input{epsf}
\usepackage{graphicx}
\usepackage{ae,aecompl}

\usepackage{hyperref}
\usepackage{amsmath}
\usepackage{amssymb}
\usepackage{mathtools}
\usepackage{bm}
\usepackage{cleveref}
\usepackage{tensor}
\usepackage{braket}
\usepackage{enumitem}
\usepackage{mhchem}
\usepackage{cancel}
\usepackage{amsthm}
\usepackage{upgreek}
\usepackage{mathrsfs}
\usepackage{slashed}

\usepackage{color}

\newcommand{\spc}{\quad \quad \quad}

\newcommand{\s}{{\mathfrak{s}}}

\def\be{\begin{equation}}
\def\ee{\end{equation}}
\def\beq{\begin{eqnarray}}
\def\eeq{\end{eqnarray}}

\theoremstyle{definition}

\theoremstyle{theorem}
\newtheorem{theorem}{Theorem}

\theoremstyle{corollary}

\begin{document}
\title{Extending Israel-Stewart theory: Causal bulk viscosity at large gradients}
\author{ L.~Gavassino}
\affiliation{
Department of Mathematics, Vanderbilt University, Nashville, TN, USA
}

\begin{abstract}
We present a class of relativistic fluid models for cold and dense matter with bulk viscosity, whose equilibrium equation of state is polytropic. These models reduce to Israel-Stewart theory for small values of the viscous stress $\Pi$. However, when $\Pi$ becomes comparable to the equilibrium pressure $P$, the evolution equations ``adjust'' to prevent the onset of far-from-equilibrium pathologies that would otherwise plague Israel-Stewart. Specifically, the equations of motion remain symmetric hyperbolic and causal at all times along any continuously differentiable flow, and across the whole thermodynamic state space. This means that, no matter how fast the fluid expands or contracts, the hydrodynamic equations are always well-behaved (away from singularities). The second law of thermodynamics is enforced exactly. Near equilibrium, these models can accommodate an arbitrarily complicated dependence of the bulk viscosity coefficient $\zeta$ on both density and temperature.
\end{abstract}

\maketitle

\section{Introduction}\label{Introne}

If a fluid expands infinitely slowly, collisional processes keep its internal degrees of freedom in equilibrium at every instant \cite[\S 11]{landau5}. In this limit, the dynamical pressure (i.e. the diagonal component of the stress tensor)  coincides with the thermodynamic pressure $P$, whose value is dictated by the equation of state \cite[\S 12]{landau5}. If, instead, the expansion rate is finite, it takes time for the internal degrees of freedom to adjust to the new size, and the fluid is driven out of equilibrium \cite{Tisza_Bulk}. Consequently, the dynamical pressure becomes $P{+}\Pi$, for some non-equilibrium correction $\Pi$ \cite[\S 81]{landau6}. The tendency of a fluid to form a non-zero $\Pi$ in response to an expansion rate is called ``bulk viscosity''.

The phenomenon of bulk viscosity has attracted increasing interest in the relativity community in recent years \cite{HuangBulk2010sa,BulkGavassino,Chabanov:2021dee,Camelio2022,CamelioSimulations2022,CeloraBulk2022nbp,GavassinoBurgers2023,GavassinoFarFromBulk2023}, due to its possible relevance to neutron star merger dynamics \cite{Sawyer_Bulk1989,Alford_2010,AlfordBulk,AlfordRezzolla,MostAlfordNoronhaBulk2021zvc,YangNoronha2023ozd}, homogeneous cosmology \cite{Weinberg1971,Salmonson1991,Otting3_1999,Maartens1995,Zimdahl1996,Fabris2006,Colistete2007}, and heavy-ion collision modeling \cite{Song:2009rh,Dusling:2011fd,PaquetDenicolBulk2015,Hydro+2018}. Together with the renewed interest, also came the realization that our standard framework for modeling bulk viscosity in relativity, the Israel-Stewart theory \cite{Israel_Stewart_1979,Hishcock1983}, is not as mathematically well-behaved as we initially thought \cite{Causality_bulk,Plumberg2022}, thereby motivating the search for new theories \cite{BemficaDNDefinitivo2020}. 

Let us briefly summarize the problem with the Israel-Stewart theory. In relativity, it is important that the equations of motion of a dissipative system do not propagate information faster than light (principle of \textit{causality} \cite{Bemfica2017TheFirst}), otherwise some violent instabilities due to backward dissipation will necessarily arise \cite{GavassinoSuperluminal2021}. In the case of IS theory with bulk viscosity, the viscous stress $\Pi$ is an independent degree of freedom of the system, and the speed of propagation of information $w$ (as measured in the fluid's local rest frame) depends on $\Pi$ as follows \cite{Causality_bulk}:
\begin{equation}
w^2=\mathcal{Z}_1 +\dfrac{\mathcal{Z}_2}{\rho {+}P{+}\Pi} \, ,
\end{equation}
where $\rho$ is the energy density, and $\mathcal{Z}_i$ are two coefficients (defined in the appendix of \cite{Causality_bulk}) that do not depend on $\Pi$. For causality, one needs that $w^2 \in [0,1]$ (note that $c \,{=}\, 1$, in our units). But then the problem is clear: If $\Pi \rightarrow -(\rho {+}P)$, the above formula diverges, and causality necessarily breaks down.

Of course, this type of causality violation occurs only when $\Pi$ is very large (and negative), and it may be tempting to argue that, in this limit, we should not trust the theory anyway. This is certainly true in principle. In practice, however, ``large-$\Pi$ situations'' are more common than one may think. For example, in heavy-ion collision simulations, causality \textit{is known} to break down at early times due to large viscous stresses \cite{Plumberg2022}. In viscous cosmology, the bulk stress $\Pi$ may become large (and negative) during fast expansion ages, like the inflation phase \cite{MaartensInflation:1996dk}, or in the proximity of cosmological singularities such as the Big Bang or the Big Rip. More importantly, hydrodynamics is known to form shocks \cite{christodoulou2017shock}, implosions \cite{MerleImplosion}, and all sorts of intrinsically non-linear phenomena (see \cite{DisconziReview2023rtt} and references therein), where $\Pi$ can become large in localized regions of space, even if the initial data is regular \cite{DisconziBreakdown2020ijk}. 

These inherently practical issues motivate the current search for theories of viscosity that remain well-behaved also at large $\Pi$ \cite{MaartensInflation:1996dk,BemficaDNDefinitivo2020}, which makes them safe for numerical implementation. It is within this set of ideas that the present article finds a place. Here, we provide a class of General Relativistic bulk viscous polytopes that ``extrapolate'' Israel-Stewart theory at large $\Pi$, in a way that avoids all the mathematical pathologies of the latter. These models are designed to survive the worst situations the Universe could throw at us. They can be initialized at arbitrarily early time in a Bjorken flow, or they can be evolved all the way to the Big Rip. No matter how hardly we strain them, they to do not give in to stress. The equations of motion have all the properties one could hope for: causality, stability, well-posedness, and strict obedience to the second law of thermodynamics. Finally, they are simple enough to be approachable analytically, but at the same time sufficiently rich to be attractive for numerical simulations.

Throughout the article, we adopt the metric signature $(-,+,+,+)$, and work in natural units $c=k_B=1$.

\newpage
\section{Overview of the model}\label{overview}
\vspace{-0.2cm}

In this section, we provide all the relevant equations of our bulk viscous polytropic fluids, and we briefly discuss their physical interpretation.

\vspace{-0.2cm}
\subsection{Constitutive relations}\label{sec:constitutiverelazia}
\vspace{-0.2cm}

Our polytropes are rheological models, which exhibit both a viscoelastic and a pseudoplastic behavior \cite{GavassinoFarFromBulk2023}. The degrees of freedom of the system are 4 fields, $\Psi(x^\alpha)\,{=}\,\{u^\mu(x^\alpha) ,\s(x^\alpha), n(x^\alpha),\Pi(x^\alpha)\}$, representing respectively the material four-velocity (with $u^\mu u_\mu=-1$), the specific entropy (i.e. the entropy per baryon), the baryon density, and the non-equilibrium part of the pressure. In the following, it will be assumed that $\s$, $n$, and $P+\Pi$ are strictly positive. To specify a model, we need to fix 4 strictly positive constant parameters, $\{m,K,\Gamma,a\}$, representing respectively the constituents' mass, the polytropic constant, the adiabatic index, and a viscosity parameter. We shall prove that consistency of the model with thermodynamics and causality requires that $\Gamma \in (1,2)$, $a \in (0,1)$, and $a{+}\Gamma{-}1 \in (0,1)$. The allowed space of values of $\{\Gamma,a\}$ is shown in figure \ref{fig:allowedaG}. 
\begin{figure}[b]
    \centering
    \includegraphics[width=0.6\linewidth]{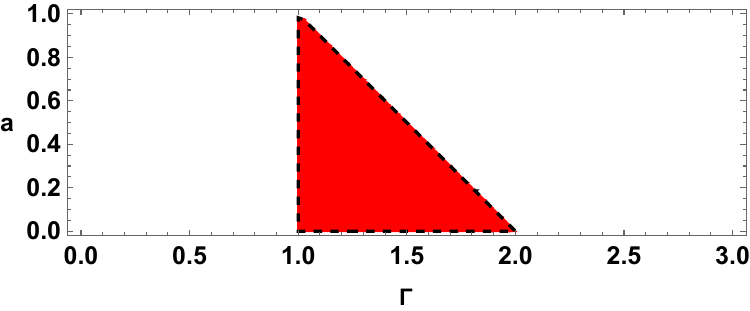}
    \caption{Values of the parameters $\{\Gamma,a \}$ that are compatible with thermodynamics and causality (the boundary is excluded). Models falling inside the red triangle are well-behaved across the entire thermodynamic phase space, i.e. whenever $\s$, $n$, and $P{+}\Pi$ are positive. Models that fall outside the red triangle may still be causal and thermodynamically consistent within some restricted range of values of $\{\s,n,\Pi \}$. However, they will break down under ``extreme conditions'' (e.g., for very large $n$ or $\Pi$).}
    \label{fig:allowedaG}
\end{figure}

With the above definitions in mind, we can now define the constitutive relations of the model. These are 4 equations, which express the pressure $P$, the energy density $\rho$, the bulk relaxation time $\tau$, and the bulk viscosity $\zeta$ in terms of the basic fields $\{\s(x^\alpha),n(x^\alpha),\Pi(x^\alpha) \}$. Such constitutive relations are provided below:
\begin{equation}\label{EoS}
\begin{split}
P={}& Kn^\Gamma \, , \\
\rho= {}& mn\big[1+e_\text{th}(\s) \big]+  \bigg[\dfrac{1}{\Gamma{-}1} + Q(\Pi/P) \bigg]P \, , \\
\tau={}& \tau(\s,n,\Pi) \, , \\
\zeta = {}&  \bigg[a+(a+\Gamma) \dfrac{\Pi}{P} \bigg] P\tau \, , \\
\end{split}
\end{equation}
where the function $Q: (-1,+\infty)\rightarrow [0,+\infty)$ and its derivative $Q'$ are given by\footnote{The argument of $Q$ must not go below $-1$, because we are assuming that $P{+}\Pi>0$, which implies $\Pi/P>-1$.}, respectively,
\begin{equation}\label{thefunctionQ}
Q(x)=\dfrac{x+\dfrac{a}{\Gamma{-}1}\bigg[ \dfrac{1}{(1{+}x)^{\frac{\Gamma-1}{a}}}-1 \bigg]}{a{+}\Gamma{-}1}\, , \quad \spc
Q'(x)=\dfrac{1-\dfrac{1}{{(1{+}x)^{\frac{a+\Gamma-1}{a}}}}}{a{+}\Gamma{-}1}\, .
\end{equation}
In equation \eqref{EoS}, the functions $e_\text{th}(\s)$ and $\tau(\s,n,\Pi)$ are arbitrary, but they should both be strictly positive and of class $C^1$ (for all $\s{>}0$, $n{>}0$, and $\Pi{>}{-}P$). Furthermore, also $e_\text{th}'(\s)$ and $e_\text{th}''(\s)$ should be strictly positive. In fact, the quantity $T=me_\text{th}'>0$ is just the temperature, while the quantity $c_v=e_\text{th}'/e_\text{th}''$ is the specific heat at constant volume. An example of an acceptable function is $e_\text{th}(\s)=b\, \s^2$, for some constant $b>0$.


\subsection{Properties of the energy density}\label{PropOfQ}
\vspace{-0.2cm}

Let us have a closer look at the formula for the energy density $\rho(\s,n,\Pi)$, as given in equation \eqref{EoS}. The equilibrium zero-temperature part, namely $mn+P/(\Gamma{-}1)$, is the usual ``textbook'' energy density of cold relativistic polytropes \cite{rezzolla_book}. The thermal correction $mne_\text{th}(\s)$ is not the most general piece one could have, but it should suffice for most purposes (in our setting, the entropy term only serves the purpose of ``storing'' the heat produced by bulk dissipation). The non-equilibrium term $PQ(\Pi/P)$ is novel and deserves in-depth explanation. 

We know from thermodynamics that the state of local equilibrium of the fluid (i.e. the state $\Pi=0$) is the state that minimizes $\rho$ at constant $n$ and $\s$ \cite[\S 5.1]{Callen_book}. Therefore, the term $PQ$ quantifies the energy cost of pushing the fluid out of local equilibrium. The shape of the function $Q(x)$ is shown in figure \ref{fig:xqx}, for some selected choices of parameters. The general features of $Q(x)$ and their physical meaning are summarized below.
\begin{itemize}
\item[(a)] The minimum energy principle requires that $Q$ be non-negative definite, and $Q=0$ only for $x=0$. This is indeed the case for the function \eqref{thefunctionQ}, provided that $\{\Gamma,a\}$ falls within the red region in figure \ref{fig:allowedaG}. Indeed, for small $x$, we have that $Q(x)=\frac{x^2}{2a}+\mathcal{O}(x^3)$, confirming that $\Pi=0$ is the minimum of $Q(\Pi/P)$, since $a>0$. 
\item[(b)] When $x{\rightarrow} -1^+$, the function $Q(x)$ diverges to $+\infty$ (again, for good choices of $\Gamma$ and $a$, see figure \ref{fig:allowedaG}). This suggests that $P+\Pi$ will ``refuse'' to become negative in these fluids, as that would require infinite energy.  
\item[(c)] If $x\rightarrow +\infty$, the function $Q(x)$ becomes infinitely large. This is reasonable, since generating infinite pressure should cost infinite energy. In particular, we have the asymptotic scaling $Q \sim x/(a{+}\Gamma{-}1)$.
\item[(d)] In the next section, we will show that consistency of the model with thermodynamics requires that $xQ'(x)\geq 0$, and $a(1{+}x)Q'+(\Gamma{-}1)Q-x=0$. Both are true as long as $\{\Gamma,a \}$ falls in the red triangle in figure \ref{fig:allowedaG}. The differential equation above constrains the form of $Q(x)$, so that equation \eqref{thefunctionQ} is forced upon us. Note that, if we changed the form $\zeta(P,\Pi)$, thermodynamics would produce a different differential equation, giving a new $Q$.
\end{itemize}

\begin{figure}[h!]
    \centering
    \includegraphics[width=0.5\linewidth]{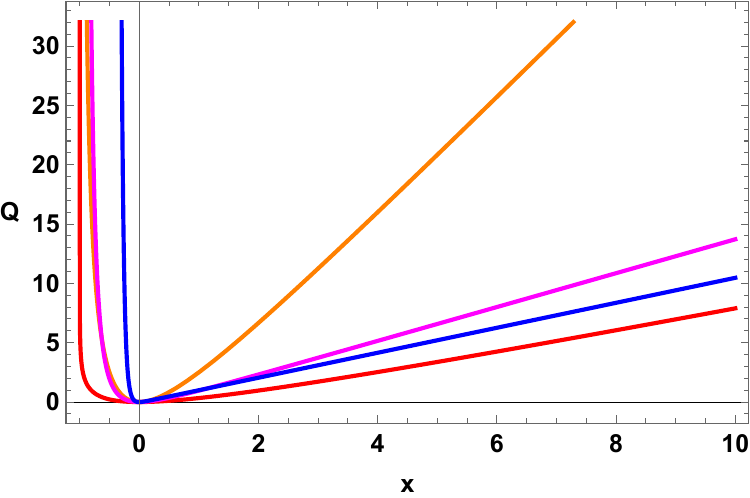}
\caption{Graph of the function $Q(x)$, defined in equation \eqref{thefunctionQ}. The quantity $PQ(\Pi/P)$ is the change in energy density associated with a non-equilibrium displacement (at fixed $n$ and $\s$). The different curves correspond to different choices of values of $\{\Gamma,a\}$, respectively, orange: $\{1.1,0.1\}$, red: $\{1.1,0.89\}$, magenta: $\{1.5,0.2\}$, blue: $\{1.9,0.05\}$.}
    \label{fig:xqx}
\end{figure}

\subsection{Equations of motion}
\vspace{-0.2cm}

Since the degrees of freedom of the fluid model are the fields $\Psi=\{u^\mu,\s,n,\Pi \}$, we need to specify how each of these fields evolves in time. We postulate the following equations of motion (defining $\Delta^{\mu \nu}=g^{\mu \nu}+u^\mu u^\nu$):
\begin{equation}\label{EoMs}
\begin{split}
u^\mu \nabla_\mu u^\nu ={}& - \dfrac{ \Delta^{\nu \mu} \nabla_\mu (P{+}\Pi)}{\rho {+}P{+}\Pi} \, , \\
u^\mu \nabla_\mu \s ={}& \dfrac{\Pi Q'(\Pi/P)}{nT\tau} \, , \\
u^\mu \nabla_\mu n ={}& - n\nabla_\mu u^\mu \, , \\
 u^\mu \nabla_\mu \Pi ={}& -\dfrac{\Pi}{\tau} -\dfrac{\zeta}{\tau} \nabla_\mu u^\mu \, . \\
\end{split}
\end{equation}
The first line is just the relativistic Euler equation, where the source of the acceleration is the \textit{total} pressure $P+\Pi$. The second line is a statement of the second law of thermodynamics, since $xQ'(x)$ is non-negative for all $x\,{>}\,{-}1$, by point (d) of section \ref{PropOfQ}. The third line is the continuity equation for baryons. Finally, the fourth line is the Israel-Stewart equation of motion for the bulk stress. We can also use the equations of motion for $\s$, $n$, and $\Pi$ to determine the evolution of $\rho$, which (thanks to the differential equation in point (d) of section \ref{PropOfQ}) matches the first law of thermodynamics:
\begin{equation}\label{energumeno}
u^\mu \nabla_\mu \rho = - (\rho{+}P{+}\Pi)\nabla_\mu u^\mu \, .
\end{equation}
We also note that the system \eqref{EoMs} is mathematically equivalent to the system of balance laws
\begin{equation}\label{BalanceLaws}
\nabla_\mu T^{\mu \nu}= 0 \, , \spc
\nabla_\mu J^\mu = 0 \, ,\spc
\tau T\nabla_\mu s^{\mu}= \Pi Q'(\Pi/P) \, ,
\end{equation}
where the stress-energy tensor $T^{\mu \nu}$, the baryon number current $J^\mu$, and the entropy current $s^\mu$ take the usual relativistic bulk-viscous form:
\begin{equation}\label{TheFluxioni}
T^{\mu \nu}= \rho u^\mu u^\nu +(P{+}\Pi)\Delta^{\mu \nu} \, , \spc
J^\mu = n u^\mu \, , \spc
s^\mu = \s n u^\mu \, . 
\end{equation}
The conservation law $\nabla_\mu T^{\mu \nu}=0$ implies that we can couple the system \eqref{EoMs} to Einstein's field equations, using $T^{\mu \nu}$ as the source of gravity. The fluid model is complete\footnote{Let us note that, if one forgets about $\s$, and just evolves $\{\rho,n,u^\mu,\Pi\}$ using \eqref{EoMs} and \eqref{energumeno}, the resulting equations have the same structure as in DNMR theory, with three differences: (a) the transport coefficients fulfill certain mathematical constraints, (b) the relaxation time $\tau$ is an arbitrary function of $\Pi$, and (c) the energy density must be initialized large enough, to be compatible with \eqref{EoS}. This implies that all current DNMR codes can be easily adjusted to evolve these polytropic models. The correspondence with DNMR also shows that the use of $\s$ as the independent variables (instead of $\rho$) was only a matter of convenience, which does not affect the physics.}.




\vspace{-0.4cm}
\subsection{A more symmetric representation}
\vspace{-0.4cm}

For later purposes, it will be convenient to rearrange the system \eqref{EoMs} in a more symmetric form. To this end, let us introduce the dimensionless quantities
$\Tilde{P}=\ln(P/P_0)$ (for some constant $P_0>0$), $\Tilde{\Pi}=\Pi/P$, and $\Tilde{\rho}=\rho/P$. Then, rather than evolving the original fields  $\Psi=\{u^\mu ,\s, n,\Pi\}$, we may evolve the dimensionless fields 
$\Tilde{\Psi}=\{u^\alpha,\s,\Tilde{P},\Tilde{\Pi}\}$. With this choice of variables, the system \eqref{EoMs} can be reorganized as follows:
\begin{equation}\label{SymmetrzedEoMs}
\begin{bmatrix}
(1{+}\Tilde{\Pi}{+}\Tilde{\rho})g\indices{^\nu _\alpha}u^\mu & 0 & (1{+}\Tilde{\Pi})\Delta^{\nu \mu} & \Delta^{\nu \mu}\\
0 & u^\mu & 0 & 0\\
(1{+}\Tilde{\Pi})\Delta\indices{^\mu _\alpha} & 0 & \dfrac{1{+}\Tilde{\Pi}}{\Gamma} u^\mu & 0\\
\Delta\indices{^\mu _\alpha} & 0 & 0 & \dfrac{u^\mu}{a(1{+}\Tilde{\Pi})}\\
\end{bmatrix}
\nabla_\mu \begin{pmatrix}
u^\alpha \\
\s \\
\Tilde{P} \\
\Tilde{\Pi}
\end{pmatrix} =
\begin{pmatrix}
0\\
\dfrac{P}{nT\tau} \Tilde{\Pi} Q'(\Tilde{\Pi})\\
0\\
 -\dfrac{\Tilde{\Pi}}{\tau a(1{+}\Tilde{\Pi})}\\
\end{pmatrix} \, .
\end{equation}


\vspace{-0.3cm}
\section{Near-equilibrium behavior}
\vspace{-0.3cm}

Before delving into the nuances of the large-$\Pi$ dynamics, let us confirm that the small-$\Pi$ limit is well-behaved. 

\vspace{-0.4cm}
\subsection{Recovering ordinary Israel-Stewart theory}
\vspace{-0.4cm}

First of all, let us show that these models are a natural ``large-$\Pi$ extrapolation'' of Israel-Stewart hydrodynamics.
To this end, we expand \eqref{EoS} in the limit of small $\Pi$, and we obtain
\begin{equation}
\begin{split}
P={}& Kn^\Gamma \, ,\\
\rho \approx {}& mn\big[1+e_\text{th}(\s) \big]+  \dfrac{P}{\Gamma{-}1} + \dfrac{\Pi^2}{2aP} \, , \\
\tau \approx {}& \tau(\s,n) \, , \\
\zeta \approx {}& a P \tau \, . \\
\end{split}
\end{equation}
These approximate constitutive relations can be used to derive the following approximate identities:
\begin{equation}
\dfrac{\partial^2 \rho}{\partial \Pi^2}\bigg|_{\Pi=0}=  \dfrac{1}{aP} \approx \dfrac{\tau}{\zeta} \, , \spc
T\nabla_\mu s^\mu \approx  \dfrac{\Pi^2}{aP\tau} \approx \dfrac{\Pi^2}{\zeta} \, . 
\end{equation}
These, together with the fourth line of \eqref{EoMs}, are the defining equations of the Israel-Stewart theory for bulk viscosity \cite{Israel_Stewart_1979,Hishcock1983,BulkGavassino,GavassinoFarFromBulk2023}, as claimed.

Let us also note in passing that these polytopic models are quite ``flexible'' when it comes to choosing $\zeta$ near equilibrium. In fact, we can always adjust the relaxation time $\tau(\s,n)$ to produce an arbitrarily complicated function $\zeta(\s,n)$, which can thus be matched with kinetic theory prescriptions.

\subsection{Stability of the equilibrium}

Since at small $\Pi$ these models become instances of Israel-Stewart theory, it should be possible to apply the Gibbs stability criterion \cite{GavassinoGibbs2021,GavassinoCausality2021,GavassinoGENERIC2022,GavassinoStabilityCarter2022} to prove that all global thermodynamic equilibria (in a fixed gravitational field) are covariantly stable \cite{GavassinoSuperluminal2021,GavassinoBounds2023}, if certain inequalities are satisfied. Indeed, we have the following theorem:
\begin{theorem}
If all the inequalities in section \ref{sec:constitutiverelazia} are obeyed, the information current, $E^\mu$, of the fluid is timelike future-directed and non-vanishing for any linear perturbation $\delta \Tilde{\Psi}\neq 0$.
\end{theorem}
\begin{proof}
The information current can be easily ``guessed'' from the observation that these polytropic models belong to the universality class $(3,1,0){-}(2,1,0)$ \cite{GavassinoUniversalityI2023,GavassinoUniversalityII2023}. Alternatively, one may apply the full machinery of the Gibbs method \cite{GavassinoGibbs2021,GavassinoCausality2021} (see appendix \ref{appInfinuzzo}). Either way, the result is 
\begin{equation}\label{EEE}
TE^\mu = \bigg[ mne_\text{th}''(\delta\s)^2+\dfrac{P}{\Gamma} (\delta \Tilde{P})^2 +\dfrac{P}{a} (\delta \Tilde{\Pi})^2+ (\rho{+}P) \delta u^\nu \delta u_\nu \bigg]\dfrac{u^\mu}{2}  +P\big( \delta \Tilde{P} +\delta \Tilde{\Pi} \big) \delta u^\mu ,
\end{equation}
where $T=me_\text{th}'(\s)$ is positive by assumption (and $\Tilde{\Pi}=0$ in equilibrium). To show that $E^\mu$ is timelike future-directed, we work in a reference frame such that
$u^\mu=(1,0,0,0)$ and $\delta u^\mu=(0,\delta u,0,0)$, and check that the quadratic form
\begin{equation}
\dfrac{2T}{P}(E^0-E^1) =(\delta u, \delta \s, \delta \Tilde{P}, \delta \Tilde{\Pi}) \begin{bmatrix}
1{+}\Tilde{\rho} & 0 & -1 & -1\\
0 & \dfrac{mne_\text{th}''}{P} & 0 & 0\\
-1 & 0 & \dfrac{1}{\Gamma} & 0\\
-1 & 0 & 0 & \dfrac{1}{a}\\
\end{bmatrix}
\begin{pmatrix}
\delta u \\
\delta \s \\
\delta \Tilde{P} \\
\delta \Tilde{\Pi} \\
\end{pmatrix}
\end{equation}
is positive definite (recall that $P>0$). To verify this, we apply Sylvester's criterion \cite[Theorem 7.2.5]{horn_johnson_1985} to the square matrix above, and require that all its principal minors $M_i$ in the lower-right corner (the ``trailing principal minors'') be strictly positive. The first three minors are just products of the diagonal elements, and they are positive since $m$, $n$, $K$, $e_\text{th}''$, $\Gamma$, and $a$ are positive. The minor $M_4$ is the determinant of the matrix itself, whose value is
\begin{equation}
M_4 = \dfrac{mne''_\text{th}}{a\Gamma P}\big[\Tilde{\rho}-(a{+}\Gamma{-}1) \big] \, . 
\end{equation}
The square bracket is positive due to the following chain of inequalities:
$\Tilde{\rho}>(\Gamma{-}1)^{-1}>1>a{+}\Gamma{-}1$.
\end{proof}

The above theorem has far-reaching implications. In fact, we can invoke the arguments of \cite{Hishcock1983,GavassinoGibbs2021,GavassinoCasmir2022}, and conclude that all global thermodynamic equilibrium states (also rotating equilibria, accelerating equilibria, and equilibria in strong gravity) are covariantly stable and thermodynamically consistent.  Furthermore, we can invoke the theorems of \cite{GavassinoCausality2021,GavassinoUniversalityI2023} and conclude that the equations of motion linearised around equilibrium are symmetric hyperbolic and causal. Finally, we can invoke \cite{MullinsHippert2023tjg,Gavassino2024Fluctuazia} and conclude that the linear theory admits a well-defined stochastic generalization, and it remains causal and stable even in the presence of random thermal fluctuations. The small-$\Pi$ limit is well-behaved.







\section{Non-linear dynamics: Consistency Theorems}

We are finally entering the central part of the paper, where the main theorems concerning the non-linear dynamics of the fluid model are proven. In the following, it will be assumed that all the conditions on $\{m,K,\Gamma,a,e_\text{th},\tau \}$ given in section \ref{sec:constitutiverelazia} are respected, and they will \textit{not} be repeated in the statement of each theorem.

\subsection{Conservation of positivity}

The theorem that follows is ultimately what makes the present model superior to ordinary Israel-Stewart theory, and it was the end goal that guided us towards equations \eqref{EoS} and \eqref{thefunctionQ} in the first place. In a nutshell, the theorem says that, if the fields take ``physically sensible'' values at $t=0$, the same will also be true at later times. 


\begin{theorem}\label{theo1} 
Fix a solution $\Psi(x^\alpha)$ of \eqref{EoMs}, and let $x^\alpha(t) \in C^1(\mathbb{R}^+)$ be a fluid-element worldline, parameterized with the proper time $t$ (i.e. $dx^\alpha/dt=u^\alpha$). Suppose that the fields $\{g_{\mu \nu}(x^\alpha), u^\mu(x^\alpha),n(x^\alpha),\Pi(x^\alpha) \}$ are of class $C^1$ in an open neighborhood of such worldline. Furthermore, suppose that $\s$, $n$, and $P{+}\Pi$ are positive at the initial event $x^\alpha(0)$. Then, $\s$, $n$, and $P{+}\Pi$ are positive along the curve $x^\alpha(t)$ for all $t\geq 0$.
\end{theorem}
\begin{proof}
First, let us prove that $n$ remains positive. Since the functions $x^\alpha(t)$ and $u^\mu(x^\alpha)$ and $g_{\mu \nu}(x^\alpha)$ are of class $C^1$ in their respective arguments, the composite function $\nabla_\mu u^\mu(t)\equiv \nabla_\mu u^\mu(x^\alpha(t))$ is of class $C^0(\mathbb{R}^+)$. It follows that the integral $\int_0^t \nabla_\mu u^\mu(t')dt'$ is finite for finite $t$. Hence, we can solve the third equation of \eqref{EoMs}, and we conclude that 
\begin{equation}\label{npositive}
n(t)=n(0)e^{-\int_0^t \nabla_\mu u^\mu(t')dt'}=n(0)\times \text{``positive number''} >0  \spc (\forall \, \, t \geq 0)\, .
\end{equation}

Let us now show that $P{+}\Pi$ remains positive. We will do this by proving that $P{+}\Pi$ can cross 0 only from below, not from above. In fact, suppose that there is some $t_0>0$ such that $P(t_0)+\Pi(t_0)=0$. Then, we can evaluate the ``combined'' equation of motion [see system \eqref{EoMs}]
\begin{equation}\label{dPpPi}
\dfrac{d}{dt} (P{+}\Pi)=-\dfrac{\Pi}{\tau}-(a{+}\Gamma)(P{+}\Pi)\nabla_\mu u^\mu 
\end{equation}
at time $t_0$, and we obtain $d (P{+}\Pi)/dt=P/\tau>0$, where we recall that $\tau$ is positive by construction. But $P+\Pi$ is of class $C^1$, so its derivative is continuous. Thus, there exists an open neighborhood $\mathcal{U}$ of $t_0$ where $d(P{+}\Pi)/dt>0$. This means that $P{+}\Pi$ is strictly increasing in the interval $\mathcal{U}$. It follows that $P{+}\Pi$ can cross 0 only if it is coming from below, i.e. if it goes from \textit{negative} to \textit{positive}, not the other way around\footnote{\textit{For mathematicians:} Here is a more formal version of the proof. Define the set $Z=\{t\geq 0 \, |\, P(t){+}\Pi(t)=0 \}$, and suppose that it is not empty. Clearly, $Z$ is bounded below, being a subset of $\mathbb{R}^+$. Therefore, it possesses an infimum $t_I \geq 0$. But $Z$ is closed, being the set of roots of a continuous function. Hence, $t_I$ belongs to $Z$, meaning that $P(t_I){+}\Pi(t_I)=0$. This also implies that $t_I>0$, because $P(0){+}\Pi(0)$ is positive. Thus, we can invoke the reasoning of the main text (taking $t_0 =t_I$) and conclude that there exists $\bar{t}\in (0,t_I)$ such that $P(\Bar{t}){+}\Pi(\Bar{t})<0$. But then, the intermediate value theorem applied to the interval $I=[0,\Bar{t}]$ tells us that there is a point in $I$ where $P{+}\Pi$ vanishes. This contradicts the definition of $t_I$ as the infimum of $Z$. It follows that $Z$ must be empty, as desired.}.

Proving that $\s$ remains positive is, at this point, straightforward. We have the following evolution equation:
\begin{equation}
\dfrac{d e_\text{th}(\s)}{dt} = \dfrac{P}{nm\tau} \, \Tilde{\Pi} Q'(\Tilde{\Pi}) \, .
\end{equation}
The right-hand side is non-negative. In fact, we know that $n$ is positive and $\Tilde{\Pi}=(P{+}\Pi)/P{-}1 $ is greater than $-1$. The function $xQ'(x)$ is designed to be non-negative whenever its argument is above $-1$. Thus, the function $e_\text{th}(\s(t))$ cannot decrease. But since $e_\text{th}(\s)$ grows with $\s$, we conclude that $\s(t)$ is non-decreasing as well.
\end{proof}

We note that, in the real world, there are substances where  $P{+}\Pi$ can become negative \cite{RajagopalCavitation2009yw,HabichCaviation2014tpa,ByresCaviaiton2020}. Modeling such fluids would require us to modify \eqref{EoS} and \eqref{thefunctionQ} in a way that the energy barrier in figure \ref{fig:xqx} be shifted more to the left.

\subsection{Dominant energy condition}

Here is a quick theorem that will be useful to prove causality.
\begin{theorem}\label{Theo2}
If $\s$, $n$, and $P{+}\Pi$ are positive, the dominant energy condition is respected, namely $\rho >|P{+}\Pi|$.
\end{theorem}
\begin{proof}
Since $P{+}\Pi$ is positive, we can drop the absolute value in the dominant energy condition, and we just need to show that $\rho>P{+}\Pi$. Equivalently, we need to show that the dimensionless function $F=(\rho{-}P{-}\Pi)/P=\Tilde{\rho}-1-\Tilde{\Pi}$ is positive. Let us write this function $F$ explicitly:
\begin{equation}
F= \dfrac{mn}{P}\big[1{+}e_\text{th}(\s) \big]+  \dfrac{2{-}\Gamma}{\Gamma{-}1} + Q(\Tilde{\Pi})-\Tilde{\Pi} \, .
\end{equation}
Using \eqref{thefunctionQ}, and rearranging the addends, we obtain
\begin{equation}
F= \dfrac{mn}{P}\big[1{+}e_\text{th}(\s) \big]+\dfrac{a(1{+}\Tilde{\Pi})^{-\frac{\Gamma{-}1}{a}}}{(a{+}\Gamma{-}1)(\Gamma{-}1)} +\dfrac{2{-}a{-}\Gamma}{a{+}\Gamma{-}1}(1{+}\Tilde{\Pi}) \, .
\end{equation}
The term proportional to $m$ is clearly positive. Moreover, each individual factor in the other two terms is positive, provided that $P+\Pi>0$, and provided that $\{\Gamma,a \}$ falls in the triangle of figure \ref{fig:allowedaG}.
\end{proof}
The dominant energy condition is a very useful property for a fluid theory to have, since it guarantees that the future development of an empty region of spacetime remains empty. In other words, the support of $T^{\mu \nu}$ cannot expand faster than light. This is known as the ``vacuum conservation theorem'', which is proven in \cite[\S 4.3]{HawkingEllis1973uf}.

\subsection{Non-linear causality}
\vspace{-0.2cm}

We are finally ready to prove that the fully non-linear system \eqref{SymmetrzedEoMs} is causal.


\begin{theorem}
If $\s$, $n$, and $P{+}\Pi$ are positive, all the characteristic speeds are real and subluminal.
\end{theorem}
\begin{proof}
The system \eqref{SymmetrzedEoMs} has the standard matrix form $\mathcal{M}^\mu \nabla_\mu \Tilde{\Psi}=\mathcal{N}$, with $\Psi=[u^\alpha,\s,\Tilde{P},\Tilde{\Pi} ]^T$, and
\begin{equation}
\mathcal{M}^\mu =    \begin{bmatrix}
(1{+}\Tilde{\Pi}{+}\Tilde{\rho})g\indices{^\nu _\alpha}u^\mu & 0 & (1{+}\Tilde{\Pi})\Delta^{\nu \mu} & \Delta^{\nu \mu}\\
0 & u^\mu & 0 & 0\\
(1{+}\Tilde{\Pi})\Delta\indices{^\mu _\alpha} & 0 & \dfrac{1{+}\Tilde{\Pi}}{\Gamma} u^\mu & 0\\
\Delta\indices{^\mu _\alpha} & 0 & 0 & \dfrac{u^\mu}{a(1{+}\Tilde{\Pi})}\\
\end{bmatrix} \, .
\end{equation}
We follow the same procedure as in \cite{Causality_bulk}, and we compute the characteristic determinant $\mathcal{P}(\xi)=\text{det}(\mathcal{M}^\mu \xi_\mu)$, where $\xi_\mu$ is an arbitrary covector. The result is
\begin{equation}
\mathcal{P}(\xi)=\dfrac{(1{+}\Tilde{\Pi}{+}\Tilde{\rho})^4}{a\Gamma} (u^\mu \xi_\mu)^5\bigg[(u^\mu \xi_\mu)^2- \dfrac{(a{+}\Gamma)(1{+}\Tilde{\Pi})}{1{+}\Tilde{\Pi}{+}\Tilde{\rho}} \Delta^{\alpha \nu}\xi_\alpha \xi_\nu\bigg] \, .
\end{equation}
Following the logic of \cite{Causality_bulk}, we conclude that, in the local rest frame of the fluid, there are 5 characteristics with vanishing speed, plus an isotropic characteristic cone with speed
\begin{equation}\label{w2speed}
w^2 = (a{+}\Gamma)   \dfrac{1{+}\Tilde{\Pi}}{1{+}\Tilde{\Pi}{+}\Tilde{\rho}} \, . 
\end{equation}
Clearly, $w^2$ it is positive, as long as $\s$, $n$, and $P{+}\Pi$ are positive (so $w$ is real). Furthermore, we can invoke the dominant energy condition, according to which $\rho>1+\Tilde{\Pi}$ (see Theorem \ref{Theo2}), and we find that
$w^2<(a{+}\Gamma)/2$, which is smaller than 1, since $\{\Gamma,a\}$ fall inside the red triangle in figure \ref{fig:allowedaG}.    
\end{proof}
More details about the scaling of $w$ with $\Tilde{\Pi}$ are given in appendix \ref{laspeedishere}.

We remark that, while the theorem was proven by treating the metric as a fixed background, one can repeat the calculation promoting $g_{\mu \nu}$ to an independent degree of freedom, which evolves according to Einstein's equations. Then, an argument by \cite{Causality_bulk} tells us that causality will hold for the whole ``Einstein+hydrodynamics'' problem.

\vspace{-0.2cm}
\subsection{Symmetric hyperbolicity}
\vspace{-0.2cm}

Our final theorem concerns symmetric hyperbolicity. Unfortunately, we cannot import Theorem 2 of \cite{Causality_bulk} directly, because one of the assumptions there is that $\zeta$ is non-negative, which may not be true for us at large $\Pi/P$. Luckily, it is straightforward to adapt the reasoning of \cite{Causality_bulk} to our case.


\begin{theorem}
If $\s$, $n$, and $P{+}\Pi$ are positive, the system \eqref{SymmetrzedEoMs} can be expressed in a symmetric hyperbolic form.
\end{theorem}
\begin{proof}
For clarity, let us work in a local Lorentz frame (so that $g_{\mu \nu}=\eta_{\mu \nu}$ and $\nabla_\mu=\partial_\mu$). We employ the normalization condition $u^\mu u_\mu=-1$ to express $u^0$ as a function of $u^j$. Furthermore, we eliminate one equation by contracting the first line of \eqref{SymmetrzedEoMs} with $\mathcal{C}^\alpha_i=g^\alpha_i+\frac{u_i}{u^0}\delta^\alpha_0$. The result is a reduced system $\mathcal{A}^\mu \partial_\mu \Upsilon=\mathcal{B}$, where $\Upsilon=[u^j,\s,\Tilde{P},\Tilde{\Pi}]^T$, and
\begin{equation}
\mathcal{A}^0=
\begin{bmatrix}
(1 {+}\Tilde{\Pi}{+}\Tilde{\rho})A_{ji}u^0 & 0 & (1{+}\Tilde{\Pi})\dfrac{u_i}{u^0} & \dfrac{u_i}{u^0} \\
0 & u^0 & 0 & 0 \\
(1{+}\Tilde{\Pi})\dfrac{u_j}{u^0} & 0 & (1{+}\Tilde{\Pi})\dfrac{u^0}{\Gamma} & 0 \\
\dfrac{u_j}{u^0} & 0 & 0 & \dfrac{u^0}{a(1{+}\Tilde{\Pi})} \\
\end{bmatrix} \, , \spc
\mathcal{A}^k=
\begin{bmatrix}
(1 {+}\Tilde{\Pi}{+}\Tilde{\rho})A_{ji}u^k & 0 & (1{+}\Tilde{\Pi})\delta^k_i & \delta^k_i \\
0 & u^k & 0 & 0 \\
(1{+}\Tilde{\Pi})\delta^k_j & 0 & (1{+}\Tilde{\Pi})\dfrac{u^k}{\Gamma}  & 0 \\
\delta^k_j & 0 & 0 & \dfrac{u^k}{a(1{+}\Tilde{\Pi})} \\
\end{bmatrix} \, , \spc
\end{equation}
with $A_{ji}=\delta_{ji}-\frac{u_j u_i}{(u^0)^2}$. This system is manifestly symmetric. In Appendix \ref{AAA}, we show that $\mathcal{A}^0$ is positive definite. Thus, the system is symmetric hyperbolic.
\end{proof}

It is well known that symmetric hyperbolicity implies local well-posedness. Hence, Theorem 5 tells us that local solutions to the initial value problem exist, are unique, and depend continuously on the initial data. 

Again, we remark that, while the above proof treats the metric as externally fixed, one can easily adapt the reasoning of \cite{Causality_bulk} to our situation, and conclude that the initial value problem will remain well-posed when coupled to the metric through Einstein's equations.

\section{Some violent flows}

It is time to put our Theorem \ref{theo1} to the test. In this section, we will consider three ``extreme'' situations, where the fluid is forced to expand very rapidly. Under these conditions, ordinary Israel-Stewart theory would drive the total pressure $P+\Pi$ below $0$, and may ultimately violate causality (when $P+\Pi \rightarrow -\rho$). Our models will not do that. Indeed, we will verify that the normalized bulk stress $\Tilde{\Pi}=\Pi/P$, which evolves according to the equation 
\begin{equation}
\tau u^\mu \nabla_\mu \Tilde{\Pi} +\Tilde{\Pi}=- a \tau(1{+}\Tilde{\Pi}) \nabla_\mu u^\mu \, ,
\end{equation}
always ``refuses'' to evolve below $-1$. 
From now on, we will assume that $\tau$ is constant, for clarity.

\vspace{-0.3cm}
\subsection{Bjorken expansion}
\vspace{-0.2cm}

For our first test, we consider a flow worldline $x^\alpha(t)$ that undergoes Bjorken expansion \cite{BjorkenFlow1983}, namely $\nabla_\mu u^\mu(t) =1/t$. Defining the normalized time $\xi=t/\tau$, the evolution equation for $\Tilde{\Pi}$ along this worldline reads
\begin{equation}\label{Bjorkuzzonzone}
\dfrac{d \Tilde{\Pi}}{d\xi}+\Tilde{\Pi}=-(1{+}\Tilde{\Pi}) \dfrac{a}{\xi} \, .
\end{equation}
At $\xi=0$, the hydrodynamic description breaks down, since $\nabla_\mu u^\mu =\infty$. Hence, we need to impose the initial data at some later time $\xi_0>0$. We assume that the fluid starts in local equilibrium, $\Tilde{\Pi}(\xi_0)=0$, and is then driven out of equilibrium by the rapid expansion. Solving \eqref{Bjorkuzzonzone} analytically, we obtain (see figure \ref{fig:BjorkenGood})
\begin{equation}\label{grione}
\Tilde{\Pi}(\xi)=ae^{-\xi} (-\xi)^{-a} \big[\Gamma(a,-\xi)-\Gamma(a,-\xi_0) \big] \, .
\end{equation}
As can be seen, no matter how early we initialize the flow (i.e. how large the gradients), $\Tilde{\Pi}$ never goes below $-1$.

\begin{figure}[h!]
    \centering
    \includegraphics[width=0.45\linewidth]{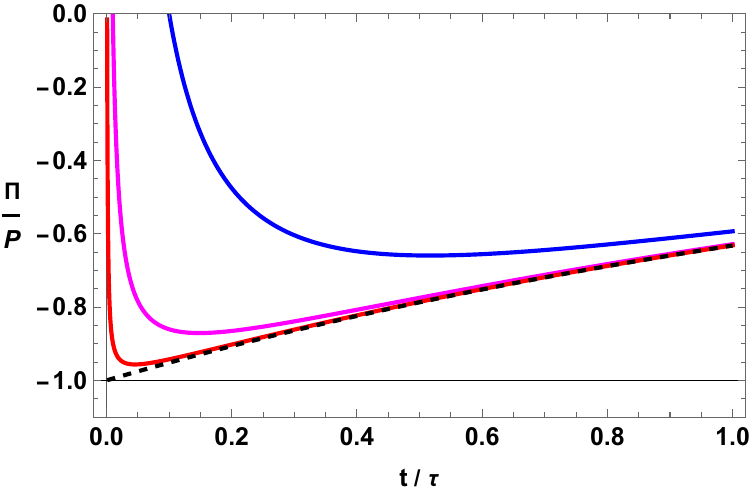}
    \caption{Evolution of the bulk pressure ratio $\Tilde{\Pi}=\Pi/P$ along a Bjorken flow, according to equation \eqref{grione}. To maximize the non-equilibrium effects, while still keeping the equations causal, we have chosen $a=0.999$. Each curve corresponds to a different choice of $\xi_0$, respectively $0.1$ (blue), $0.01$ (magenta), and $0.001$ (red). The dashed line corresponds to the limit where $\xi_0 \rightarrow 0^+$. As can be seen, the value of $\Pi/P$ never goes below $-1$, in accordance with Theorem \ref{theo1}.}
    \label{fig:BjorkenGood}
\end{figure}

\vspace{-0.3cm}
\subsection{Fast Hubble expansion}
\vspace{-0.2cm}

For the second test, we consider a flow worldline $x^\alpha(t)$ at rest in a FLRW spacetime, with constant Hubble parameter $H>0$, so that $\nabla_\mu u^\mu(t)\, {=}\, 3H$ \cite{Maartens1995}. Introducing the dimensionless time $\xi=t/\tau$ and the dimensionless expansion rate $\theta=3aH\tau$, the evolution equation for $\Tilde{\Pi}$ along this worldline reads 
\begin{equation}
\dfrac{d\Tilde{\Pi}}{d\xi}+ \Tilde{\Pi} = -(1{+}\Tilde{\Pi}) \theta \, .
\end{equation}
Setting $\Tilde{\Pi}(0)=0$, we can solve this equation analytically, and we obtain (see figure \ref{fig:deSitterGood})
\begin{equation}\label{deSitt}
\Tilde{\Pi}(\xi)= \dfrac{\theta}{1{+}\theta} \big[e^{-(1+\theta)\xi}-1 \big]\, .
\end{equation}
As can be seen, even in the limit of infinitely fast expansion ($\theta \rightarrow +\infty$), $\Tilde{\Pi}$ does not go below $-1$.
\newpage
\begin{figure}[h!]
    \centering
    \includegraphics[width=0.47\linewidth]{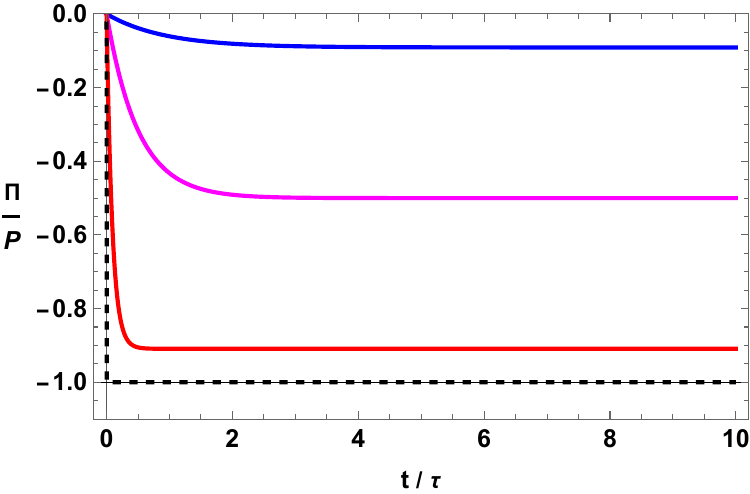}
    \caption{Evolution of the bulk pressure ratio $\Tilde{\Pi}=\Pi/P$ driven by a Hubble expansion, according to equation \eqref{deSitt}.  Each curve corresponds to a different choice of $\theta$, respectively $0.1$ (blue), $1$ (magenta), and $10$ (red). The dashed line corresponds to the limit where $\theta \rightarrow +\infty$. Again, the value of $\Pi/P$ never goes below $-1$, in accordance with Theorem \ref{theo1}.}
    \label{fig:deSitterGood}
\end{figure}

\subsection{Big Rip singularity}

For our last test, we consider a flow worldline $x^\alpha(t)$ that is approaching a ``Big-Rip'' singularity, namely a location where $\nabla_\mu u^\mu \rightarrow +\infty$  in finite proper time. We model this phenomenon as a simple divergence $a\nabla_\mu u^\mu(t)=(\tau-t)^{-1}$, so the singularity is located at $t=\tau$. Then, setting again $\xi=t/\tau$, we have the differential equation
\begin{equation}
\dfrac{d\Tilde{\Pi}}{d\xi}+ \Tilde{\Pi} = -\dfrac{1{+}\Tilde{\Pi}}{1-\xi}\, .
\end{equation}
Taking, as initial data, $\Tilde{\Pi}(0)=0$, we obtain (see figure \ref{fig:BigRip})
\begin{equation}\label{BigRip}
\Tilde{\Pi}(\xi)= -e^{-\xi } (\xi -1) \big[-e \text{Ei}(\xi -1)+e \text{Ei}(-1)+1\big]-1 \, .
\end{equation}
As can be seen, $\Tilde{\Pi}$ remains above $-1$ all the way to the singularity.

\begin{figure}[h!]
    \centering
    \includegraphics[width=0.47\linewidth]{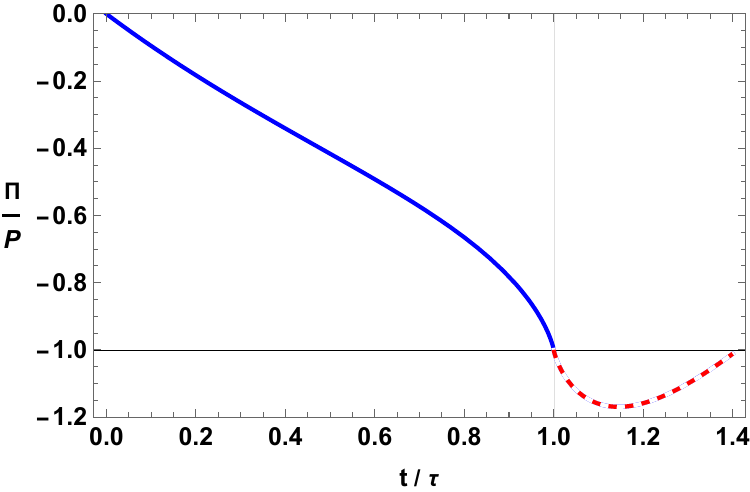}
    \caption{Evolution of the bulk pressure ratio $\Tilde{\Pi}=\Pi/P$ as it approaches the Big-Rip singularity, according to equation \eqref{BigRip}. Consistently with Theorem \ref{theo1}, $\Pi/P$ remains above $-1$ all the way to the singularity (which is located at $t/\tau=1$). When the singularity is reached, the fields are no longer $C^1$, and the theorem no longer applies. What comes next (red dashed) is unphysical, since the spacetime itself has ended.}
    \label{fig:BigRip}
\end{figure}

\newpage
\section{Viscous heating at large $\Pi$}
\vspace{-0.3cm}

We conclude the paper with an additional consistency check. We showed earlier that these polytopes obey the second law of thermodynamics also at large $\Pi$. This means that, if we make the fluid undergo a cycle of large compression and expansion, the system will heat up. Let us confirm that this is indeed true with a quick example. 

We consider a flow worldline $x^\alpha(\xi)$ (with $\xi=t / \tau$) that undergoes a periodic cycle of transformations, characterized by the following time-dependent density and expansion rate:
\vspace{-0.2cm}
\begin{equation}\label{violuzzo}
n(\xi)= \frac{1}{[2+\cos (\xi )]^3} \, , \spc \tau \nabla_\mu u^\mu (\xi) =-\dfrac{3\sin(\xi)}{2+\cos(\xi)} \, .
\end{equation}
This process is fast, since its timescale equals the relaxation time $\tau$. But it is also very ``intense'', since $n(\xi)$ oscillates between a maximum of $1$ and a minimum of $1/27\sim 0.03$.

We solve the equations of motion taking $\{m,K,\Gamma,a\}=\{1,1,4/3,1/3 \}$ and setting $\s(0)=0$. For this choice of parameters, the system admits an analytical solution (see figure \ref{fig:Heating}):
\vspace{-0.2cm}
\begin{equation}\label{CheLungone}
\begin{split}
P(\xi)={}& \frac{1}{[2+\cos (\xi )]^4} \, , \\
\Tilde{\Pi}(\xi)={}& \frac{\sin (\xi )-\cos (\xi )}{4+2 \cos (\xi )} \, , \\
\Tilde{\rho}(\xi)={}& \frac{-3 [\sin (\xi )+\sin (2 \xi )-42]+82 \cos (\xi )+8 \cos (2 \xi )+2 \sqrt{3} [\cos (\xi )+2]^2 \text{arcot} \left(\sqrt{3} \cot \left(\frac{\xi }{2}\right)\right)}{12 [2+\cos (\xi )]} \, , \\
\end{split}
\end{equation}
where the function ``$\text{arcot}$'' needs to be periodically shifted upward of an amount equal to $\pi$, to make $\Tilde{\rho}$ continuous. The figure agrees with the expectations: The net ``$(P{+}\Pi)dV$'' work over a cycle is non-zero, resulting in a net increase in thermal energy at the end of each cycle.

\begin{figure}[h!]
    \centering
    \includegraphics[width=0.43\linewidth]{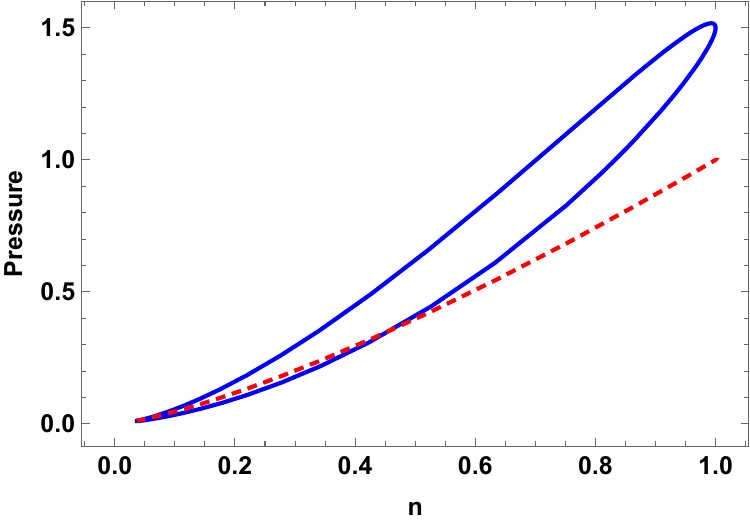}
    \includegraphics[width=0.43\linewidth]{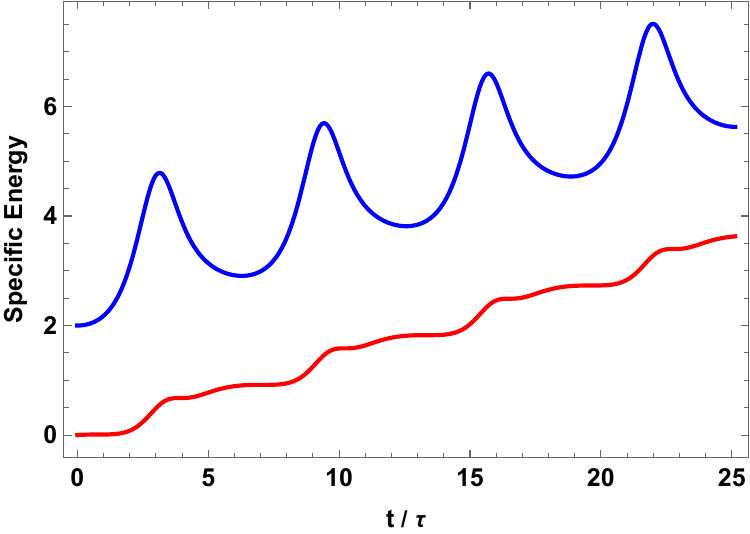}
\caption{Evolution of pressure and energy along a violent cycle of compression and expansion, according to equation \eqref{CheLungone}. Left panel: The total pressure $P{+}\Pi$ (blue) undergoes a hysteresis loop, and differs considerably from the equilibrium part $P$ (red). Right panel: After every cycle, the energy per particle $\rho/n$ (blue) is larger. This is a manifestation of the second law of thermodynamics, according to which $\s$, and thus the thermal energy $me_\text{th}(\s)$ (red), is non-decreasing in time.}
    \label{fig:Heating}
\end{figure}

\section{Conclusions}
\vspace{-0.4cm}

We have built a class of polytropic bulk-viscous fluid models that ``refuse'' to become acausal, even under extremely large expansion rates. The basic idea is simple: We created an infinitely tall energy barrier in the proximity of $\Pi=-P$ (see figure \ref{fig:xqx}), which ``repulses'' $\Pi$, keeping the fluid above the causality threshold (see figures \ref{fig:BjorkenGood}, \ref{fig:deSitterGood}, and \ref{fig:BigRip}). 

These models are ready to be simulated numerically. In fact, their initial value problem is locally well-posed at arbitrary density, temperature, and (positive) $P+\Pi$, also when coupled to Einstein's equations. The evolution equations also have the advantage of being fully irreversible, since they obey the second law of thermodynamics \textit{exactly}, unlike DNMR or BDNK, which respects the second law only approximately.

From a thermodynamic perspective, these polytopic models were constructed with cold and dense matter in mind. This manifests itself in the assumption that the equilibrium pressure does not depend on the temperature. In an upcoming paper, we will develop similar models for ``hot'' matter at zero chemical potential. Such models may find application in heavy-ion collision simulations.

\section*{Acknowledgements}

This work is partially supported by a Vanderbilt's Seeding Success Grant. I thank Marcelo Disconzi for reading the manuscript and providing useful feedback.

\appendix

\section{Computation of the information current}\label{appInfinuzzo}

We follow the same procedure as in \cite{GavassinoGENERIC2022,GavassinoStabilityCarter2022,GavassinoUniversalityII2023}. Specifically, we fix a background metric $g_{\mu \nu}$, and an external heat bath, with fugacity $\alpha^\star$ (which is a constant) and inverse-temperature four-vector $\beta^{\star \nu}$ (which is a Killing vector). Then, we consider a one-parameter family of states $\Psi(\epsilon)$, where $\epsilon=0$ is the (yet to be discovered) global equilibrium state of the fluid in contact with this bath. Then, we compute the first and second derivatives with respect to $\epsilon$ of the vector field $\phi^\mu=s^\mu +\alpha^\star J^\mu+\beta^\star_\nu T^{\mu \nu}$, where the hydrodynamic fluxes are functions of $\Psi(\epsilon)$ through \eqref{TheFluxioni}, while $\alpha^\star$, $\beta^\star_\nu$ and the metric $g_{\mu \nu}$ are held constant in the differentiation (i.e. they do not depend on $\epsilon$). Denoting by ``$\, \cdot \,$'' the differentiation in $\epsilon$ (not in time!), we have the following:
\begin{equation}
\begin{split}
\phi^\mu ={}& (K n^\Gamma{+}\Pi)\beta^{\star \mu} +\bigg[\s n +\alpha^\star n +u^\nu \beta^\star_\nu \bigg( mn(1{+}e_\text{th})+\dfrac{\Gamma Kn^\Gamma}{\Gamma{-}1} +\Pi +PQ    \bigg) \bigg] u^\mu \, ,\\
\dot{\phi}^\mu ={}& \big(\Gamma K n^{\Gamma-1}\dot{n}{+}\dot{\Pi}\big)\beta^{\star \mu} +\bigg[ n \dot{\s}+\s \dot{n} +\alpha^\star \dot{n} +\dot{u}^\nu \beta^\star_\nu \bigg( mn(1{+}e_\text{th})+\dfrac{\Gamma Kn^\Gamma}{\Gamma{-}1} +\Pi +PQ    \bigg) \\
+{}& u^\nu \beta^\star_\nu \bigg( m\dot{n}(1{+}e_\text{th})+mne_\text{th}'\dot{\s}+\dfrac{\Gamma^2 Kn^{\Gamma-1}}{\Gamma{-}1} \dot{n}+\dot{\Pi} +\dot{[PQ]}    \bigg) \bigg] u^\mu \\ 
+{}& \bigg[\s n +\alpha^\star n +u^\nu \beta^\star_\nu \bigg( mn(1{+}e_\text{th})+\dfrac{\Gamma K n^\Gamma}{\Gamma{-}1} +\Pi +PQ    \bigg) \bigg] \dot{u}^\mu \, ,\\
\ddot{\phi}^\mu ={}& \big(\Gamma K n^{\Gamma-1}\ddot{n}{+}\Gamma(\Gamma{-}1)Kn^{\Gamma-2}\dot{n}^2{+}\ddot{\Pi}\big)\beta^{\star \mu}\\+{}& \bigg[ n \ddot{\s}+\s \ddot{n}+2\dot{\s}\dot{n} +\alpha^\star \ddot{n} +\ddot{u}^\nu \beta^\star_\nu \bigg( mn(1{+}e_\text{th})+\dfrac{\Gamma Kn^\Gamma}{\Gamma{-}1} +\Pi +PQ    \bigg) \\
+{}& 2\dot{u}^\nu \beta^\star_\nu \bigg( m\dot{n}(1{+}e_\text{th})+mne_\text{th}'\dot{\s}+\dfrac{\Gamma^2 Kn^{\Gamma-1}}{\Gamma{-}1} \dot{n}+\dot{\Pi} +\dot{[PQ]}    \bigg)\\
+{}& u^\nu \beta^\star_\nu \bigg( m\ddot{n}(1{+}e_\text{th})+2m\dot{n}e_\text{th}'\dot{\s}+mne_\text{th}''\dot{\s}^2+mne_\text{th}'\ddot{\s}+\Gamma^2 Kn^{\Gamma-2} \dot{n}^2+\dfrac{\Gamma^2 Kn^{\Gamma-1}}{\Gamma{-}1} \ddot{n}+\ddot{\Pi} +\ddot{[PQ]}    \bigg) \bigg] u^\mu  \\+{}& 2\bigg[ n \dot{\s}+\s \dot{n} +\alpha^\star \dot{n} +\dot{u}^\nu \beta^\star_\nu \bigg( mn(1{+}e_\text{th})+\dfrac{\Gamma Kn^\Gamma}{\Gamma{-}1} +\Pi +PQ    \bigg) \\
+{}& u^\nu \beta^\star_\nu \bigg( m\dot{n}(1{+}e_\text{th})+mne_\text{th}'\dot{\s}+\dfrac{\Gamma^2 Kn^{\Gamma-1}}{\Gamma{-}1} \dot{n}+\dot{\Pi} +\dot{[PQ]}    \bigg) \bigg] \dot{u}^\mu \\ 
+{}& \bigg[\s n +\alpha^\star n +u^\nu \beta^\star_\nu \bigg( mn(1{+}e_\text{th})+\dfrac{\Gamma K n^\Gamma}{\Gamma{-}1} +\Pi +PQ    \bigg) \bigg] \ddot{u}^\mu \, .\\
\end{split}
\end{equation}
By definition, the equilibrium state $\Psi(0)$ is the state that maximizes the thermodynamic potential $\Phi=\int_\Sigma \phi^\mu d\Sigma_\mu$ across an arbitrary Cauchy surface $\Sigma$. This implies that $\dot{\phi}^\mu(0)$ should vanish for all possible $\dot{\Psi}(0)$. With a bit of algebra, one finds that this happens if and only if
\begin{equation}\label{Adue!}
\beta^\star_\nu =\dfrac{u_\nu}{T} \, , \spc T\alpha^\star=m(1{+}e_\text{th})+\dfrac{\Gamma K n^{\Gamma-1}}{\Gamma{-}1}-T\s \, , \spc \Pi=0 \, .
\end{equation}
The first condition is the usual statement that $u^\nu/T$ is a Killing vector in equilibrium \cite{BecattiniBeta2016} (a.k.a. the ``Tolman law''). The second relation is the ``Klein law'', namely $\mu/T=\text{const}$, where the chemical potential is the quantity on the right-hand side. The third condition is the obvious statement that global equilibrium implies local equilibrium.

If we evaluate the formula for $\ddot{\phi}$ at $\epsilon=0$, and use \eqref{Adue!}, we obtain
\begin{equation}
-\dfrac{1}{2} T\ddot{\phi}^\mu(0) = \bigg[ mne_\text{th}''\dot{\s}^2+\Gamma Kn^{\Gamma-2} \dot{n}^2 +\dfrac{\dot{\Pi}^2}{aKn^\Gamma}+ \bigg( mn(1{+}e_\text{th})+\dfrac{\Gamma Kn^\Gamma}{\Gamma{-}1}    \bigg) \dot{u}^\nu \dot{u}_\nu \bigg]\dfrac{u^\mu}{2}  +\big( \Gamma K n^{\Gamma-1} \dot{n}     +\dot{\Pi} \big) \dot{u}^\mu .
\end{equation}
Now, we recall that the information current is just $E^\mu =-\ddot{\phi}(0) \epsilon^2/2$, and that the linear perturbation to an arbitrary observable $A$ is given by $\delta A=\dot{A}(0)\epsilon$. The final result (after changing basis from $\delta \Psi$ to $ \delta \Tilde{\Psi}$) is equation \eqref{EEE}.

\section{Speed of information}\label{laspeedishere}

The speed of propagation of information $w$ is given by equation \eqref{w2speed}. The function $\Tilde{\rho}$ contains the term $mn(1+e_\text{th})/P$, which is always positive, and thus reduces the overall value of $w$. At high densities (for fixed values of $\s$ and $\Tilde{\Pi}$), this term disappears, since it scales like $1/n^{\Gamma-1}$, and $\Gamma>1$. Hence, we have that 
\begin{equation}
w^2 \leq w^2(n{=}\infty)= \frac{a+\Gamma-1}{1+\frac{a (1+\Tilde{\Pi})^{-\frac{a+\Gamma-1}{a}}}{(\Gamma-1) (a+\Gamma)}} \, ,
\end{equation}
which is smaller than 1, as long as $\{\Gamma,a\}$ falls inside the red triangle. This high-density speed is plotted in figure \ref{fig:Wona}.

\begin{figure}[h!]
    \centering
    \includegraphics[width=0.5\linewidth]{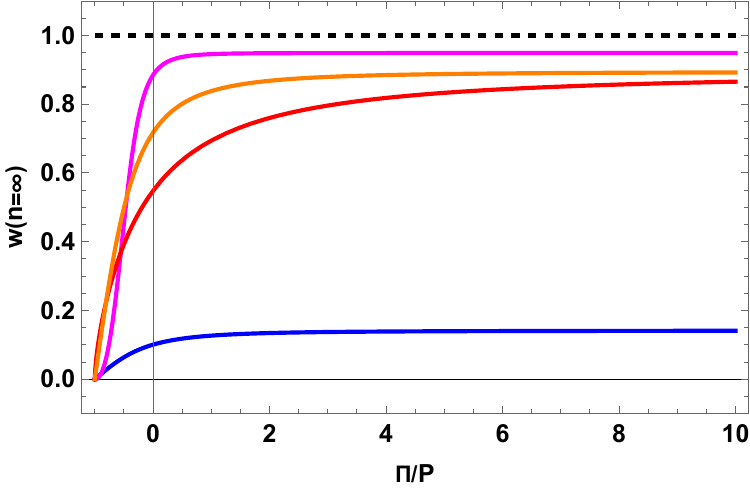}
    \caption{High-density speed of information $w(n=\infty)$ as a function of $\Tilde{\Pi}=\Pi/P$ for different values of $\{\Gamma,a \}$, respectively, blue:$\{1.01,0,01 \}$, magenta:$\{ 1.7,0.2\}$, red:$\{1.2,0.6 \}$, and orange:$\{1.4,0.4 \}$. The dashed line is the light-speed limit.}
    \label{fig:Wona}
\end{figure}

\section{Positive definiteness of $\mathcal{A}^0$}\label{AAA}

It is straightforward to see that the definiteness of $\mathcal{A}^0$ is invariant under rotations. Hence, we can arrange the axes so that $u^\mu=(\gamma,\gamma v,0,0)$, where $\gamma$ and $v$ are the Lorentz factor and speed of the fluid. Then, the matrix $\mathcal{A}^0$ becomes
\begin{equation}
\mathcal{A}^0=
\begin{bmatrix}
(1 {+}\Tilde{\Pi}{+}\Tilde{\rho})\gamma^{-1} &0 &0 & 0 & (1{+}\Tilde{\Pi})v & v \\
0 & (1 {+}\Tilde{\Pi}{+}\Tilde{\rho})\gamma & 0 &0  & 0 & 0 \\
0 & 0 & (1 {+}\Tilde{\Pi}{+}\Tilde{\rho})\gamma & 0 & 0 & 0 \\
0 & 0 & 0 & \gamma & 0 & 0 \\
(1{+}\Tilde{\Pi})v & 0& 0& 0 & (1{+}\Tilde{\Pi})\dfrac{\gamma}{\Gamma} & 0 \\
v & 0&  0 & 0 & 0 & \dfrac{\gamma}{a(1{+}\Tilde{\Pi})} \\
\end{bmatrix} \, .
\end{equation}
We can demonstrate positive definiteness by verifying that all the leading principal minors are positive, again by Sylvester's criterion \cite[Theorem 7.2.5]{horn_johnson_1985}. These are given by
\begin{equation}
\begin{split}
M_n ={}& (1 {+}\Tilde{\Pi}{+}\Tilde{\rho})^n\gamma^{n-2}  \spc (n=1,2,3)\, , \\
M_4 ={}&  (1 {+}\Tilde{\Pi}{+}\Tilde{\rho})^3\gamma^2 \, , \\
M_5 ={}& (1 {+}\Tilde{\Pi}{+}\Tilde{\rho})^3 \gamma^{3} \dfrac{1{+}\Tilde{\Pi}}{\Gamma} \bigg[1- \dfrac{\Gamma}{a{+}\Gamma} w^2 v^2 \bigg]  \, , \\
M_6 ={}& \dfrac{(1{+}\Tilde{\Pi}{+}\Tilde{\rho})^3\gamma^4 (1{-}w^2  v^2)}{a\Gamma}  \, . \\
\end{split}
\end{equation}
Recalling that $0<w^2<1$ by causality, and that $\{\Gamma,a \}$ falls within the red triangle of figure \ref{fig:allowedaG}, it is easy to see that all these minors are indeed positive.

\newpage
\bibliography{Biblio}

\label{lastpage}

\end{document}